\documentclass[runningheads]{llncs}
\usepackage{amsmath} 
\usepackage{graphicx}
\usepackage[T1]{fontenc}
\usepackage{pxfonts}
\def\calP{\mathcal{P}}
\def\calL{\mathcal{L}}
\newtheorem{observation}{Observation}
\makeatletter
\renewcommand\subsubsection{\@startsection{subsubsection}{3}{\z@}%
                       {-18\p@ \@plus -4\p@ \@minus -4\p@}%
                       {0.5em \@plus 0.22em \@minus 0.1em}%
                       {\normalfont\normalsize\bfseries\boldmath}}
\makeatother
\usepackage{fullpage}
\usepackage{doi}

\begin{document}

\title{The Bichromatic Two-Center Problem on Graphs\thanks{This research was supported in part by U.S. National Science Foundation under Grant CCF-2339371.}
\thanks{A preliminary version of this paper will appear in \textit{Proceedings of the 18th Annual International Conference on Combinatorial Optimization and Applications (COCOA 2025)}.}}

\author{Qi Sun\and Jingru Zhang}
\authorrunning{Q. Sun and J. Zhang}

\institute{Cleveland State University, Cleveland, Ohio 44115, USA\\
\email{q.sun1@vikes.csuohio.edu, j.zhang40@csuohio.edu}}

\maketitle  
\begin{abstract}
    In this paper, we study the (weighted) bichromatic two-center problem on graphs. 
    The input consists of a graph $G$ of $n$ (weighted) vertices and $m$ edges, and a set $\calP$ 
    of pairs of distinct vertices, where no vertex appears in more than one pair. 
    The problem aims to find two points (i.e., centers) on $G$ by assigning vertices 
    of each pair to different centers so as to minimize 
    the maximum (weighted) distance of vertices to their assigned centers 
    (so that the graph can be bi-colored with this goal).   
    To the best of our knowledge, this problem has not been studied on graphs, including tree graphs. 
    In this paper, we propose an $O(m^2n\log n\log mn)$ algorithm 
    for solving the problem on an undirected graph provided with the distance matrix, 
    an $O(n\log n)$-time algorithm for the problem on trees, 
    and a linear-time approach for the unweighted tree version. 
\end{abstract}

\section{Introduction}
The bichromatic two-center problem is a new facility location problem, 
proposed by Arkin et al. in~\cite{ref:ArkinBi15}, with the goal to 
construct two public transport hubs to connect origin and destination sites while 
aiming to minimize the maximum transportation cost between the sites and the hubs. 
To the best of our knowledge, all previous work has focused on the 
planar unweighted case~\cite{ref:ArkinBi15,ref:WangIm22}, 
and no result is known for the problem on graphs. 
In this paper, we present the first algorithms for the (weighted) bichromatic 
two-center problem on various graphs.

Let $G = (V,E)$ be an undirected graph with $n$ vertices and $m$ edges 
where each vertex $v$ is of a weight $w(v)\geq 0$ 
and each edge $e$ is of length $l(e)>0$. In addition, the input includes 
a set of $k$ pairs of distinct vertices so that no vertex of $G$ occurs in 
more than one pair. Let $\calP = \{P_1, \cdots, P_k\}$ represent this set 
of pairs; denote by $v_i$ and $u_i$ the two vertices in each pair $P_i$. 

For any two vertex $u,v$, denote $e(u,v)$ by the edge connecting them. 
We talk about `points' on $e(u,v)$ so that a point $q$ of $e(u,v)$ is 
at distance $t(q)$ to $u$ along (segment) $e(u,v)$. So, the distance of 
two points $q_1$ and $q_2$ of $G$ is the length of their shortest path 
$\pi(q_1,q_2)$. 

For each pair $P_i\in\calP$, define $\phi_i(q_1,q_2)$ as the minimized maximum (weighted) 
distance by assigning $v_i$ and $u_i$ to different points of $q_1$ and $q_2$; 
namely, $\phi_i(q_1,q_2) $ is the minimum value of $\max\{w(v_i)d(v_i,q_1), w(u_i)d(u_i, q_2)\}$ and $\max\{w(u_i)d(u_i,q_1), w(v_i)d(v_i, q_2)\}$. 
The bichromatic two-center problem aims to find two points $q^*_1$ and $q^*_2$ on $G$ 
so as to minimize $\max_{1\leq i\leq k}\phi_i(q_1, q_2)$. 

Consider coloring vertices assigned to $q^*_1$ red, and vertices assigned to $q^*_2$ blue 
in an optimal solution. In addition, the center of vertices in a subset of $V$ is a point on $G$ 
that minimizes the maximum (weighted) distance from it to the vertices. 
Hence, this problem is indeed a graph $2$-coloring variant that aims to choose from each pair 
a vertex to be red and let the other to be blue, in such a way that minimizes the maximum 
one-center objective values respectively for red and blue vertices. 

In this paper, we propose an $O(nm^2\log n\log mn)$-time algorithm for the problem 
on an undirected graph, provided with the distance matrix. 
Moreover, an $O(n\log n)$-time algorithm is given for solving the problem on a tree, 
and an $O(n)$-time improvement is obtained for the unweighted tree case. 

\paragraph{Related work.} As motivated by a chromatic clustering problem~\cite{ref:DingSo11} arising in biology and 
transportation, Arkin et al.~\cite{ref:ArkinBi15} proposed the bichromatic two-center 
problem for pairs of points in the plane, which aims to choose from each pair 
a point to be colored red and let the other to be blue so as to minimize the maximum of 
the two radii of the minimum balls respectively enclosing red and blue points. 
They gave an linear-time algorithm for the problem in $L_\infty$ metric, and 
an $O(n^3\log n)$-time exact algorithm as well as two approximation algorithms 
for the problem in $L_2$ metric. Recently, Wang and Xue~\cite{ref:WangIm22} improved their 
exact algorithm for the $L_2$ metric to $O(n^2\log n)$ and gave a faster approximation 
algorithm. 
As far as we are aware, however, no previous result exists for the problem on graph networks. 

This bichromatic two-center problem is closely related to the one-center and two-center problems~\cite{ref:WangIm22}. These two problems have been studied on various graphs in the literature. 
Megiddo~\cite{ref:MegiddoLi83} proposed a prune-and-search algorithm for addressing 
the one-center problem on trees in linear time. By adapting Megiddo's prune-and-search scheme, 
Ben-Moshe et al.~\cite{ref:Ben-MosheAn06} solved the two-center problem on tree in linear time. 
On cactus graphs, Ben-Moshe et al.~\cite{ref:BenMosheEf07} first gave a `binary search' algorithm 
to solve the one-center problem and then extended their one-center scheme to address the two-center problem in $O(n\log^3n)$ time. 
On general graphs, Kariv and Hakimi~\cite{ref:KarivAnC79} gave the first exact algorithms for the one-center problem and the unweighted case. Later, Bhattacharya and Shi~\cite{ref:BhattacharyaIm14} improved their result to $O(m^2n\log^2 n)$ by proposing an $O(m^2n\log n)$-time feasibility test.  

To solve our problem, a geometric piercing problem needs to be addressed: Given a set of subsets of two-sided axis-parallel rectangles, the goal is to find a point in the plane that pierces (or lies in) one rectangle in each subset. 
To the best of our knowledge, no exact algorithm exists for this problem, and most exact works focus on finding $O(1)$ points to pierce $n$ geometric objects, e.g., disks and rectangles~\cite{ref:LofflerLa10,ref:CarmiSt2023,ref:SharirRe96,ref:NussbaumRe97}. 

\paragraph{Our approach.}
Let $\lambda^*$ be the optimal objective value, and denote by $q_1^*$ and $q_2^*$ the two centers. 
We observe that $\lambda^*$ is determined by the \textit{local} center of two vertices 
on an edge of $G$. Because the local center of any two vertices on an edge 
is determined by an intersection of their distance functions $w(v)d(v,x)$ w.r.t. $x$ on the edge. 
By forming the distance function for each vertex w.r.t. each edge of $G$, the 
line-arrangement searching technique can be adapted to find $\lambda^*$ with assistance 
of our feasibility test which determines for any given value $\lambda$ whether 
$\lambda\geq\lambda^*$. 

To determine whether $\lambda\geq\lambda^*$, we determine for 
every pair of edges whether each edge contains a point such that the objective 
value w.r.t. the two points is at most $\lambda$. This problem can be reduced into 
the above-mentioned geometric piercing problem. The sweeping technique is adapted to 
solve this problem in $O(n\log n)$ time with the assistance of our data structure. 

For $G$ being a tree, we perform a `binary search' on it to find the two edges that respectively 
contain $q^*_1$ and $q^*_2$ with the assistance of our key lemma: Given any point of the tree, 
we can find its \textit{hanging} subtrees (i.e., two subtrees incident to the point) 
respectively containing $q^*_1$ and $q^*_2$ in linear time. This key lemma is based on our 
linear-time feasibility test. As a result, $\lambda^*$ can be computed in $O(n\log n)$ time. 

When every vertex of the tree is of same weight, we observe that $\lambda^*$ is 
determined by the longest path(s) between two vertices assigned to $q^*_1$ or $q^*_2$, 
where one of its endpoints must be the endpoint of the longest path of the tree. 
This implies that the optimal assignment for each pair $P_i$ is determined by 
$\phi_i(\alpha,\beta)$ where the path between vertices $\alpha$ and $\beta$ 
is a longest path of the tree, which leads a linear-time algorithm.  

\paragraph{Outline.} In Section~\ref{sec:preliminary}, we introduce some observations and review a technique. We present our algorithm for the problem on general graphs in Section~\ref{sec:graph}, and our algorithms for the weighted and unweighted cases on tree graphs in Section~\ref{sec:tree}.

\section{Preliminaries}\label{sec:preliminary}
Let $V(\calP)$ be the set of vertices in all pairs of $\calP$. Note that $V(\calP)$ might be a subset of $V$. If so, we set the weight of each 
vertex in $V/V(\calP)$ as zero and then form new pairs for all vertices in $V/V(\calP)$. 
It should be mentioned that when $n$ is odd, 
before forming new pairs, we join a new vertex into $G$ by connecting it and any vertex of $G$ via an edge of a positive length. By the provided distance matrix, the distance of 
this new vertex to every other vertex can be known in $O(1)$ time.  

Let $\calP'$ denote the set of all obtained pairs including those in $\calP$, 
and $G'$ denote the new graph in the case that $n$ is odd. 
Clearly, solving the problem on $G$ w.r.t. $\calP$ is 
equivalent to solving the problem on $G'$ w.r.t. $\calP'$. The time complexity of the reduction is $O(n+m)$. Below, we thus focus on solving the case of $V(P) = V$. 

For any subset $V'$ of $V$ and any edge $e$ of $G$, 
there is a point on $e$ that minimizes the maximum (weighted) distance of it to vertices in $V'$; 
such a point is the \textit{local} center of $V'$ on $e$. The center of $V'$ on $G$ is the one with 
the smallest maximum (weighted) distance, i.e., the smallest one-center objective value of $V'$, among all its local centers. 

\begin{figure}
\begin{minipage}{0.32\linewidth}
    \includegraphics[width = 0.8\linewidth]{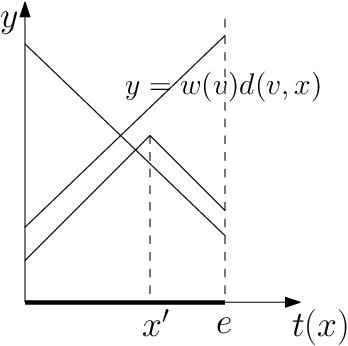}
    \caption{Illustrate the three cases of function $y = w(v)d(v,x)$ w.r.t. $x$ on edge $e$. Point $x'$ is the (unique) semicircular point of $v$ on $e$. }
    \label{fig:distance}
\end{minipage}
\begin{minipage}{0.65\linewidth}
    \includegraphics[width=0.95\linewidth]{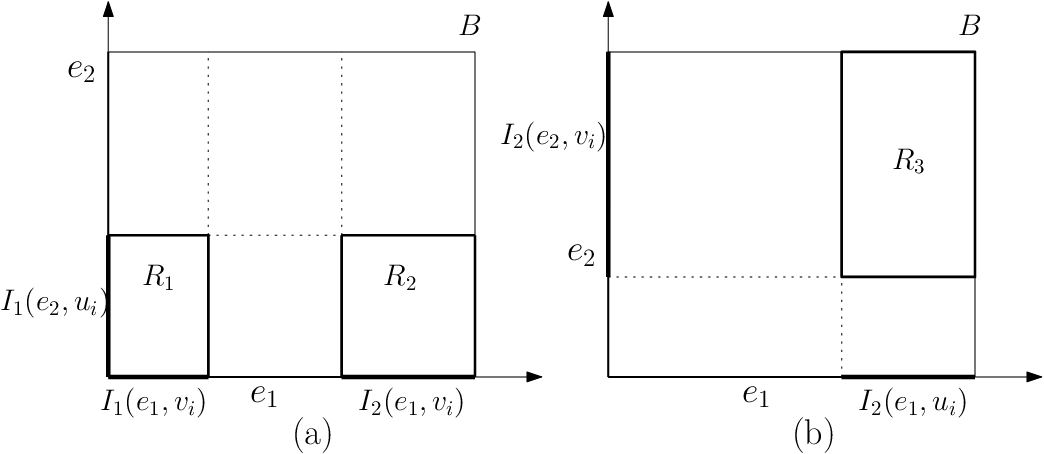}
    \caption{Illustrate that each assignment of pair $P_i=(v_i, u_i)$ leads at most four disjoint rectangles in box $B$ each sharing a vertex of $B$. In (a), $v_i$ (resp., $u_i$) is assigned to center $q_1\in e_1$ (resp., $q_2\in e_2$) so two rectangles $R_1$ and $R_2$ are generated; (b) illustrates the other assignment which leads rectangle $R_3$.}
    \label{fig:rectangleset}
\end{minipage}
\end{figure}

Consider the weighted distance $w(v)d(v,x)$ of each $v\in V'$ to a point $x$ of $e$. Recall that $x$ is at distance $t(x)$ to an incident vertex of $e$ along $e$. 
As $t(x)$ increases (namely, $x$ moves along $e$ from an incident vertex to the other), $w(v)d(v,x)$ linearly increases or decreases, and otherwise, it first 
linearly increases until $x$ reaches $v$'s \textit{semicircular} point on $e$, 
which is the unique point on $e$ so that $v$ has two shortest paths to it containing different incident vertices of $e$, and then linearly decreases. So, $w(v)d(v,x)$ is a piecewise linear function w.r.t. $t(x)$ (see Fig.~\ref{fig:distance}). As shown in~\cite{ref:KarivAnC79,ref:BhattacharyaIm14} 
for the graph one-center problem, the local center of $V'$ on $e$ 
is determined by an intersection between distance functions 
$y = w(v)d(v,x)$ of two vertices, where their slopes are of opposite signs. Hence, we have the following observation. 

\begin{observation}\label{obs:graphsetoflambda}
    $\lambda^*$ belongs to the set of values $w(v)d(v,x) = w(u)d(u,x)$ for each pair of vertices $u,v\in V$ w.r.t. every edge of $G$. 
\end{observation}

Clearly, this candidate set is of size $O(mn^2)$ and it can be computed in $O(mn^2)$ time by the 
provided distance matrix of $G$. 

When $G$ is a tree, for any two vertices, there is only one (simple) path between them. 
For any path $\pi'$, supposing $x$ is a point on $\pi'$, $w(v)d(v,x)$ is a 
convex function w.r.t. the distance of $x$ to an endpoint of $\pi'$. So, the center of $V'$ is determined by the center of two vertices on their path, and the size of the candidate set is decreased to $O(n^2)$~\cite{ref:MegiddoLi83,ref:BanikTh16,ref:WangAn21}. 

\begin{observation}\label{obs:treesetoflambda}
    For $G$ being a tree, $\lambda^*$ belongs to the set of values 
    $w(v)d(v,x) = w(u)d(u,x)$ for each pair of vertices $u,v\in V$ w.r.t. $x$ on the path between the two vertices. 
\end{observation}

Besides of computing the candidate set, an implicit way to enumerate this candidate set is forming the distance function $y = w(v)d(v,x)$ for each vertex w.r.t. every edge of $G$ in the $x,y$-coordinate plane such that $\lambda^*$ belongs to the set of 
the $y$-coordinates of intersections between the distance functions. 

Consider the set of $O(mn)$ lines obtained by extending each line segment on the graph 
of every distance function formed in the $x,y$-coordinate plane. These lines define a plane subdivision where each intersection of lines 
is a vertex and vice versa. A benefit for implicitly enumerating the candidate set 
is that the line arrangement search technique~\cite{ref:ChenAn13} can be utilized to search for $\lambda^*$ among the $y$-coordinates of all intersections without 
computing the arrangement of the lines in advance. Note that both explicit and implicit ways need the assistance of a feasibility test that determines for any given value $\lambda$, whether $\lambda\geq\lambda^*$. The line arrangement search technique is reviewed as follows. 

Suppose $L$ is a set of $N$ lines in the plane. Denote by $\mathcal{A}(L)$ the arrangement of lines in $L$. For any point $p$ in the plane, denote by $x(p)$ its $x$-coordinate and $y(p)$ its $y$-coordinate.  
Let $v_1(L)$ be the lowest vertex of $\mathcal{A}(L)$ whose $y(v_1(L))$ is a feasible value, and let $v_2(L)$ be the highest vertex of $\mathcal{A}(L)$ with $y(v_2(L)) < y(v_1(L))$. 
By the definition, $y(v_2(L))<\lambda^*\leq y(v_1(L))$, and no vertex in $\mathcal{A}(L)$ has a $y$-coordinate in the range $(y(v_2(L)), y(v_1(L)))$. 
Lemma~\ref{lem:linearrangement} was proved in~\cite{ref:ChenAn13} for finding the two vertices. 

\begin{lemma}\label{lem:linearrangement}
     \textup{\cite{ref:ChenAn13}} Both vertices $v_1(L)$ and $v_2(L)$ can be computed in $O((N+\tau)\log N)$ time, where $\tau$ is the time for a feasibility test. 
\end{lemma}

\section{The problem on a graph}\label{sec:graph}
In this section, we first give the algorithm for computing $\lambda^*$ and present our feasibility test later in Section~\ref{subsec:graphdecision}. 

The above analysis in Section~\ref{sec:preliminary} leads our algorithm: For each edge $e(u,u')$ of $G$, we form the distance function $y = w(v)d(v,x)$ for every vertex $v\in V$ w.r.t. $x$ on $e(u,u')$ in the $x,y$-coordinate plane where $e(u,u')$ is on the positive side of the $x$-axis with $u$ at the origin. By the given distance matrix, this can be carried out in $O(mn)$ time. Next, we extend each line segment on the graph of every distance function into a line, which generates in $O(mn)$ time a set $L$ of $O(mn)$ lines.   

Consider the line arrangement $\mathcal{A}$ of lines in $L$. We say that a value is feasible 
iff it is no less than $\lambda^*$, and otherwise, it is infeasible. 
Because $\lambda^*$ equals the $y$-coordinate of the lowest vertex in $\mathcal{A}$ whose $y$-coordinate is feasible. We apply Lemma~\ref{lem:linearrangement} to find this lowest vertex of $\mathcal{A}$ by using Lemma~\ref{lem:graphfeasibility} for the feasibility test. The proof of Lemma~\ref{lem:graphfeasibility} is given in Section~\ref{subsec:graphdecision}. Thus, the time complexity for computing $\lambda^*$ is $O(m^2n\log n\log mn)$. 

\begin{lemma}\label{lem:graphfeasibility}
    Given any value $\lambda>0$, we can decide whether $\lambda\geq\lambda^*$ in $O(m^2n\log n)$ time. 
\end{lemma}

\begin{theorem}\label{the:graph}
    The bichromatic two-center problem can be solved in $O(m^2n\log n\log mn)$ time. 
\end{theorem}

\subsection{Proof of Lemma~\ref{lem:graphfeasibility}}\label{subsec:graphdecision}
Given any value $\lambda>0$, the feasibility test aims to determine whether there are 
two points $q_1$ and $ q_2$ on $G$ so that $\max_{1\leq i\leq k}\phi_i(q_1, q_2)\leq\lambda$. 
If yes, then $\lambda\geq\lambda^*$ so it is feasible, 
and otherwise, $\lambda<\lambda^*$ and it is infeasible. 

Our strategy is to determine for every pair of edges, 
say $(e_1, e_2)$, whether there exist two points $q_1\in e_1$ and $q_2\in e_2$ 
so that $\max_{1\leq i\leq k}\phi_i(q_1, q_2)\leq\lambda$. 
Note that $e_1$ might be same as $e_2$ for considering the case where $q^*_1$ 
and $q^*_2$ are on the same edge. Refer to the problem of determining the existence of such two points for a pair of edges as a \textit{local} feasibility test. 

Below, we shall present how to solve a local feasibility test: Given two edges $e_1, e_2\in G$, our goal is to find two points $q_1\in e_1$ and $q_2\in e_2$ such that 
by assigning one vertex of each pair $P_i$ to $q_1$ and the other to $q_2$, 
the maximum (weighted) distance of each vertex to its assigned point is no more than $\lambda$. 

We first present below that solving the local feasibility test is equivalent to solving the following geometric piercing problem: Let $U$ be a set of axis-parallel rectangles in a box each sharing a vertex of the box, 
i.e., a set of quadrants. Let $U_1, \cdots, U_k$ be a set of subsets of $O(1)$ rectangles in $U$. The goal is to determine whether there is a point that pierces (or lies in) a rectangle in $U_i$ for each $1\leq i\leq k$. 

\paragraph{Problem reduction.} For each vertex $v$, solving $w(v)d(v,x)\leq\lambda$ 
for $x\in e_i$ with $1\leq i\leq 2$ generates at most two disjoint closed intervals 
each of which includes an incident vertex of $e_i$; we say that each obtained interval is 
a \textit{feasible} interval of $v$ on $e_i$. 
Suppose $e_i$ is incident to vertex $\alpha_i$ and $\beta_i$. 
Let $I_1(v, e_i)$ and $I_2(v, e_i)$ represent $v$'s feasible intervals on $e_i$ 
where $I_1(v, e_i)$ includes $\alpha_i$ and $I_2(v, e_i)$ includes $\beta_i$; 
note that $I_1(v, e_i)$ or $I_2(v, e_i)$ may be null if only one interval 
or none is obtained by solving $w(v)d(v,x)\leq\lambda$ for $x\in e_i$. 

Consider the $x,y$-coordinate plane where $e_1$ (resp., $e_2$) is on the positive side 
of $x$-axis (resp., $y$-axis) with $\alpha_1$ (resp., $\alpha_2$) at the origin. 
Let $B$ be the intersection box (rectangle) of slabs 
$\{(x,y)|0\leq x\leq x(\beta_1)\}$ and $\{(x,y)|0\leq y\leq y(\beta_2)\}$.

For each pair $P_i=(v_i,u_i)$ of $\calP$, consider the intersection 
of (vertical) slabs respectively containing $v_i$'s feasible intervals $I_1(e_1, v_i)$ and $I_2(e_1, v_i)$ on $e_1$ 
and (horizontal) slabs respectively containing $u_i$'s feasible intervals $I_1(e_2, u_i)$ 
and $I_2(e_2, u_i)$ on $e_2$; their intersection consists of four disjoint axis-parallel 
rectangles in $B$, each sharing a vertex with $B$. 
As illustrated in Fig.~\ref{fig:rectangleset} (a), there exist two points $q_1\leq e_1$ and $q_2\leq e_2$ 
so that $v_i$ is covered by $q_1$ and $u_i$ is covered by $q_2$ (under $\lambda$) 
iff point $(x= x(q_1), y = y(q_2))$ pierces one of the intersection rectangles. 

Consider the other assignment where $v_i$ is assigned to the center on $e_2$ 
and $u_i$ to the center on $e_1$. Likewise, two points $q_1\leq e_1$ and $q_2\leq e_2$ exist 
such that $w(v_i)d(v_i, q_2)\leq\lambda$ and $w(u_i)d(u_i,q_1)\leq\lambda$ 
iff point $(x = x(q_1), y = y(q_2))$ pierces one of $O(1)$ disjoint two-sided intersection 
rectangles of slabs respectively containing $v_i$'s feasible intervals 
$I_1(e_2, v_i)$ and $I_2(e_2, v_i)$ on $e_2$ and horizontal slabs 
respectively containing $u_i$'s feasible intervals $I_1(e_1, u_i)$ 
and $I_2(e_1, u_i)$ on $e_1$ (see Fig.~\ref{fig:rectangleset} (b)).

It should be mentioned that if $v_i$ and $u_i$ have no feasible intervals on $e_1$ (resp., $e_2$), 
then it is not likely to find two points $q_1\in e_1$ and $q_2\in e_2$ 
with $\phi_{i}(q_1,q_2)\leq\lambda$ and thereby, $\lambda$ is infeasible. 

Let $U_i$ be the set of the $O(1)$ axis-parallel rectangles obtained for the two 
assignments of $v_i$ and $u_i$. 
The above analysis implies that there are two points $q_1\leq e_1$ 
and $q_2\leq e_2$ with $\phi_i(q_1, q_2)\leq\lambda$ iff point $(x = x(q_1), y = y(q_2))$ 
pierces a rectangle in $U_i$ for each $1\leq i\leq k$. 
It follows that if there is a point in $B$ that pierces a rectangle of $U_i$ 
for each $1\leq i\leq k$, then $\lambda$ is feasible, and otherwise, 
$\lambda$ is infeasible w.r.t. the two edges $e_1$ and $e_2$. 

As a result, each local feasibility test can be reduced to an instance of 
the geometric piercing problem, which asks whether there is a point in $B$ 
that pierces at least one rectangle of $U_i$ for each $1\leq i\leq k$. 
Regarding to the reduction time, the feasible intervals of each vertex on any edge can be 
obtained in $O(1)$ time with the provided distance matrix; additionally, each 
rectangle set $U_i$ can be obtained in $O(1)$ time. Hence, the reduction can be done in $O(n)$ time. 

The geometric piercing problem is quite interesting since each geometric object 
is a set of multiple axis-parallel two-sided rectangles. 
To the best of our knowledge, no previous work exists for addressing this problem. 
We shall present below a sweeping algorithm for solving it in $O(n\log n)$ time. 

\subsubsection{Solving the geometric piercing problem} 
We say that a rectangle set $U_i$ is hit by a point if it pierces 
a rectangle of $U_i$. Hence, the goal is to find a point (in box $B$) that 
hits set $U_i$ for each $1\leq i\leq k$. 

Obviously, if all sets $U_i$'s can be hit by a point, 
then at least one such point is on the left side of a rectangle. Thus, our strategy is to sweep the plane by a 
vertical line $h$ from left to right while determining whether 
such a point exists at each event where $h$ encounters a rectangle. 

Let $X = \{x_1, \cdots, x_N\}$ be the $x$-coordinate sequence of 
vertical sides of $B$ and all rectangles in $U$ from left to right; 
clearly, $N = O(n)$. 
For each $1\leq j\leq N$, let $S_j$ be the intersection segment of 
line $x = x_j$ and box $B$ where $S_j$'s upper endpoint is denoted by $a_j$ 
and its lower endpoint is by $b_j$. 
For each $1\leq i\leq k$, the intersection of line $x = x_j$ and (rectangles of) set $U_i$ is null, a closed segment on $S_j$ 
that includes $a_j$ or $b_j$, or two disjoint closed segments on $S_j$ 
which end respectively at $a_j$ and $b_j$. 
A point piercing an intersection segment of line $x=x_j$ and 
each set $U_i$ hits all sets $U_i$. 

Instead, we consider the following problem to determine whether 
such a point on $S_j$ exists or not. 
Let $l_i$ be the complement segment of the intersection segment of $S_j$ and $U_i$. 
$l_i$ could be either of the following: $S_j$, an open segment on $S_j$, 
a half-open segment closing at either $a_j$ or $b_j$, and null. 
See Fig.~\ref{fig:piercing} for an example. 
We have the following observation. 

\begin{figure}
\begin{minipage}{0.48\linewidth}
    \centering
    \includegraphics[width=0.55\linewidth]{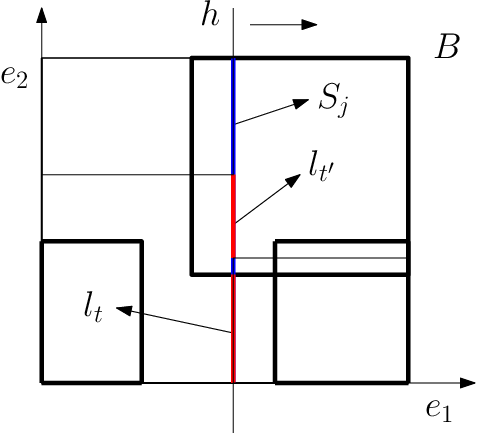}
    \caption{Illustrating Observation~\ref{obs:hitting}. $U_t$ has three (bold) rectangles and $U_{t'}$ has two (non-bold) rectangles. The sweeping line $h$ intersects $B$ at the (blue) line segment $S_j$. The complement segments of $U_t$ and $U_{t'}$ on $S_j$ are $l_t$ and $l_{t'}$.}
    \label{fig:piercing}
\end{minipage}
\begin{minipage}{0.48\linewidth}
    \centering
    \includegraphics[width=0.7\linewidth]{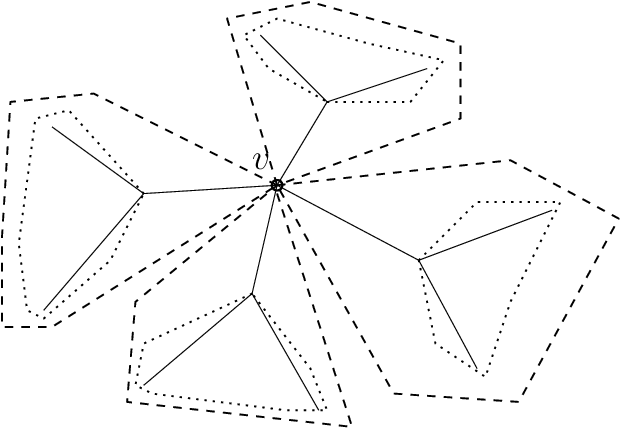}
    \caption{Vertex $v$ contains four split subtrees (marked by dotted lines) and four hanging subtrees (marked by dashed lines); $v$ is a virtual point in its each hanging subtree.}
    \label{fig:splitsubtree}
\end{minipage}
\end{figure}

\begin{observation}\label{obs:hitting}
    $\cup_{i=1}^{k} l_i = S_j$ iff no point of $S_j$ hits sets $U_i$ for all $1\leq i\leq k$. 
    If $\cup_{i=1}^{k} l_i\neq S_j$, then an endpoint of a complement segment 
    hits all sets $U_i$'s. 
\end{observation}
\begin{proof}
    Suppose no point on $S_j$ hits all rectangle sets $U_i$. For any point $q$ on $S_j$, there must be at least one set, say $U_{i'}$, so that $q$ is not contained in its intersection segment(s) with $S_j$. It follows that every point of $S_j$ is contained in a complement segment. 
    Thus, $\cup_{i=1}^{k} l_i = S_j$. For the other direction, suppose $\cup_{i=1}^{k} l_i = S_j$. Every point of $S_j$ is contained in a complement segment. 
    It follows that no point on $S_j$ hits sets $U_{i}$ for all $1\leq i\leq t$. 
    This proves the first statement. 

    If $\cup_{i=1}^{k} l_i \neq S_j$, an endpoint of the intersection segment(s) 
    between a rectangle set and $S_j$ hits all sets $U_i$; this endpoint is also an endpoint 
    of the complement segment of the rectangle set on $S_j$. Hence, the second statement is true. \qed
\end{proof}

Observation~\ref{obs:hitting} implies that for each $1\leq j\leq N$, 
determining whether there is a point on $S_j$ that hits sets $U_i$ for 
all $1\leq i\leq k$ is equivalent to determine 
whether $\cup_{i=1}^{n} l_i = S_j$. 

In fact, for each $1\leq j\leq N$, the complement segment $l_i$ of set $U_i$ 
on $S_j$ is the intersection of line $x=x_j$ (or $S_j$) and the complement area 
of rectangles of $U_i$ in box $B$. Because each $U_i$ consists of at most 
eight axis-parallel rectangles in $B$, each sharing a vertex with $B$. 
The complement area of $U_i$ in $B$ consists of $O(1)$ disjoint rectilinear polygons 
each with $O(1)$ sides such that all their sides except those in $B$'s boundary 
are open and no vertical line intersects more than one of them. 
Consequently, as sweeping $h$ from left to right, there will be $O(1)$ distinct complement segments 
for each $U_i$ on the intersection segment between $B$ and $h$. 
Denote by $\calL$ the set of the $O(1)$ complement segments of set $U_i$ for all $1\leq i\leq k$. 
The above analysis and observation~\ref{obs:hitting} imply the following observation. 

\begin{observation}\label{obs:endpointhitting}
     If all sets $U_i$ can be hit by a point of $B$, then an endpoint of a segment in $\calL$ must hit all sets $U_i$. 
\end{observation}

There is no point in box $B$ hitting all sets $U_i$ iff every point of $B$ is in the complement area 
of a rectangle set. It follows that no point of $B$ hits all sets $U_i$ 
iff the measure of the union of all complement polygons of sets $U_i$ is equal to $B$'s. 
It seems that the algorithms~\cite{ref:BentleyAl79,ref:PreparataCo85} 
for the 2D Klee's measure problem, which asks for the measure of the union of axis-parallel rectangles, can be adapted to compute the measure of their union after decomposing every polygon into $O(1)$ axis-parallel rectangles. 
However, every complement polygon in our problem might be open at its sides. Hence, the 2D Klee's measure algorithms cannot be adapted to solve our problem. 

Now we are ready to describe our sweeping algorithm for solving this 
geometric piercing problem. 

\paragraph{Preprocessing.} We first compute the sequence $X$ as well as sequences $Y = \{y_1, \cdots, y_M\}$, $I_{in} = \{I^1_{in}, \cdots, I_{in}^N\}$, and $I_{out} = \{I^1_{out}, \cdots, I^N_{out}\}$, where $Y$ is the ascending order of $y$-coordinates of endpoints of segments in $\calL$, and each $I^{t}_{in}$ (resp., $I^{t}_{out}$) stores the indices of all rectangle sets $U_i$ that contain 
rectangles with the left (resp., right) side on line $x=x_j$. 
Clearly, $M = O(n)$, $\cup_{t=1}^N |I^t_{in}| = O(k)$ and $\cup_{t=1}^N |I^i_{out}| = O(k)$; the four sequences can be obtained in $O(n\log n)$ time. 

Furthermore, we construct a balanced binary search tree $\Gamma$ on $Y$ to maintain 
the event status when the sweeping line $h$ intersects a line $x = x'$: 
letting $S'$ be the intersection segment between $h$ and box $B$, the number 
of complement segments on $S'$ of sets $U_i$ that contain the point on $S'$ 
of $y$-coordinate equal $y_t$ for each $y_t\in Y$, 
and the minimum of such attributes for all points defined by $Y$ on $S'$. 
(Because a point of $y$-coordinate in $Y$ must hit all sets $U_i$ 
if all sets $U_i$ can be hit by a point of $B$.)

In details, all $y$-coordinates of $Y$ are stored in the leaf level 
as keys from left to right in order. Each leaf node $u$ stores another 
attribute $\eta$, which is the number of complement segments on $S'$ 
including the point on $S'$ of $y$-coordinate equal to its key, 
denoted by $key(u)$. Every internal node $u$ is attached two attributes: 
The smallest key $key(u)$ and the smallest $\eta$ 
among all nodes in the subtree $\Gamma_{u}$ rooted at $u$. All attributes $\eta$ are initialized as zero. 

In addition, $\Gamma$ supports two operations: The \textit{range-addition} operation which is adding a given value to attributes $\eta$ in leaves of keys within a given range, and the \textit{minimum} query for the minimum $\eta$ over the whole tree (or all leaves). 
To achieve a logarithmic-time updating, we apply the lazy propagation mechanism~\cite{ref:HalimCo13,ref:IbtehazMu21} 
to $\Gamma$ (so as to avoid updating all nodes of keys within the given range at once): 
Each node $u$ contains an extra attribute $z(u)$ such that when $z(u)\neq 0$, an addition of $z(u)$ needs to be applied to attribute $\eta$ of each (leaf) node in $\Gamma_u$ 
(and $u$ is a \textit{lazy} node); $z(u)$ is set as zero at the beginning. 

Clearly, $\Gamma$ is of size $O(n)$ and height $O(\log n)$; 
provided with $Y$, the construction can be finished in $O(n)$ time. 
In the following, to distinguish the status of $\Gamma$ at different events, we denote by $\Gamma(x')$ the status of $\Gamma$ when $h$ intersects a line $x=x'$. 

\paragraph{The algorithm.} Sweep the plane by the vertical line $h$ from left to right. 
For each $x_j\in X$, $h$ halts at the event of $h$ intersecting line $x = x_j$. 
For each event, we first determine the status $\Gamma(x_j)$ of $\Gamma$ and 
then decide whether $\cup_{i=1}^{k} l_i= S_j$ by performing a minimum query 
on $\Gamma(x_j)$. If the obtained value is zero, then a point of $S_j$ 
is not contained in any complement segment on $S_j$, so a point of $S_j$ hits all sets $U_i$; 
otherwise, by observation~\ref{obs:hitting}, we have $\cup_{i=1}^{k} l_i= S_j$. 
More details of handling each event are as follows: 

\begin{itemize}
    \item $h$ intersects line $x= x_1$: $I^1_{out}= \emptyset$. 
    For each index $i\in I_{in}^1$, we first compute the complement segment $l_i$ of $U_i$ on $S_1$ 
    in constant time, and then perform a range-addition operation on $\Gamma$ to add $1$ to attribute $\eta$ of every leaf of key in the $y$-coordinate range of $l_i$ (namely, inserting a new segment that covers each point on $S_1$ of $y$-coordinate in the $y$-coordinate range of $l_i$). Clearly, computing $\Gamma(x_1)$ takes $O(|I_{in}^1|\log n)$ time. Afterward, we perform a minimum query on $\Gamma(x_1)$ in $O(1)$ time. If it is zero, then a point on $S_1$ hits all sets $U_i$ so we terminate the algorithm, 
    and otherwise, we continue the sweeping. 
    
    It should be mentioned that $\Gamma(x_1)$ is consistent with the status $\Gamma(x')$ for any $x_1<x'<x_2$. 

    \item $h$ intersects line $x= x_j$ with $1<j< N$. $\Gamma$'s current status is consistent 
    with its status for $h$ in slab $x_{j-1}<x< x_j$. 
    % so we first determine the status $\Gamma(x_j)$. 
    Recall that only set $U_t$ for each $t\in I^j_{in}$ contains rectangles whose left sides are in line $x=x_j$ (namely, $h$ enters some rectangles of $U_t$ at line $x=x_j$). 
    To determine $\Gamma(x_j)$, for each $t\in I^j_{in}$, 
    we compute the complement segment $l'_t$ of $U_t$ on a vertical line $x = x'$ with $x_{j-1}<x'<x_j$, 
    as well as its complement segment $l_t$ on $S_j$. Next, for each $t\in I^j_{in}$, we first 
    decrement attributes $\eta$ of leaves with keys in the $y$-coordinate range of $l'_t$ 
    and then increment attributes $\eta$ in leaves with keys in the $y$-coordinate range of $l_t$, 
    which can be carried out in $O(\log n)$ time by performing two range-addition operations on $\Gamma$. 
    This derives $\Gamma(x_j)$ and the running time is $O(|I_{in}^j|\log n)$. 
    
    We proceed to perform a minimum query on $\Gamma$ to determine whether 
    $\cup_{i=1}^{k} l_i= S_j$. If the obtained value equals zero, 
    then a point on $S_j$ hits all sets $U_i$ so we immediately return. 
    Otherwise, we proceed to determine the status $\Gamma(x')$ for any $x_j<x'<x_{j+1}$ 
    by processing each set $U_t$ with $t\in I_{out}^j$. (Recall that each set $U_t$ with $t\in I_{out}^j$ contain a rectangle whose right side is in line $x=x_j$ so $h$ is leaving a rectangle.) 
    
    For each $t\in I_{out}^j$, we compute the complement segment $l_t$ of set $U_t$ on $S_j$, 
    and perform a range-addition operation to decrement attribute $\eta$ of 
    each leaf of key in the $y$-coordinate range of $l_t$; next, we compute its complement segment $l''_t$ 
    on the intersection segment between $B$ and line $x=x'$ for any $x'\in (x_j, x_{j+1})$, 
    followed by executing a range-addition operation to increment 
    attribute $\eta$ of each leaf of key in the $y$-coordinate range of $l''_t$. 

    The overall running time for handling such event is $O((|I_{in}^j|+|I_{out}^j|)\log n)$.

    \item $h$ intersects line $x = x_N$: Because $\Gamma(x_N)$ is consistent with $\Gamma(x')$ for any $x_{N-1}<x'<x_N$, which is the current status of $\Gamma$. We perform a minimum query on $\Gamma$. If it is zero then a point on $S_N$ hits sets $U_i$. Otherwise, no point in box $B$ (including its boundaries) 
    hits sets $U_i$ (so the union of complement rectilinear polygons of sets $U_i$ in $B$ is exactly box $B$). 
\end{itemize}

It remains to show how to perform on $\Gamma$ a range-addition operation with lazy propagation and a query for the minimum attribute $\eta$ of $\Gamma$. 

\begin{itemize}
    \item Range-addition. Given any range $\gamma$ and value $c$, 
    the goal is to add value $c$ to attribute $\eta$ of each leaf 
    with key in $\gamma$ while updating attribute $\eta$ of necessary internal nodes. 
    Let $u_a$ and $u_b$ be the two leaves whose keys are the boundary values of $\gamma$ with $key(u_a)\leq key(u_b)$. 
    First, we locate the lowest common ancestor $lcs(u_a, u_b)$ of $u_a$ and $u_b$, where the searching path for $u_a$ and $u_b$ starts to split. During the searching for $lcs(a,b)$, 
    for each node $u$ encountered, we check if $z(u)\neq 0$. If so, then an addition of $z(u)$ needs to be applied to attribute $\eta$ of all nodes in subtree $\Gamma_u$; but with the lazy propagation, 
    we `push down' this update only to its children by adding $z(u)$ to attribute $\eta$ of $u$ and attribute $z(u')$ of its each child $u'$; last, we reset $z(u)$ of $u$ as zero.

    Once we arrives at $lcs(u_a,u_b)$, we continue to search for $u_a$ in the subtree rooted at its left child $u_L$, 
    and then find $u_b$ in the one rooted at its right child $u_R$. 
    During the traversal on $\Gamma_{u_L}$, for each node, we first push down the update to its children and then decide 
    the searching direction for $u_a$. Note that after arriving $u_a$, every node on the searching path of $u_a$ is not lazy. This means that the current value $z(u')$ in any lazy child node $u'$ of nodes on this path is the sum of all previous addends to attribute $\eta$ of every node in the subtree rooted at $u'$; 
    it follows that $z(u')$ plus the current value of $\eta$ in $u'$ is the up-to-date value of $\eta$ 
    after all prior updates.  
    
    Furthermore, during the traceback from $u_a$ to $u_L$, we perform an addition of value $c$ to necessary nodes in range $\gamma$ (with lazy propagation). 
    More specifically, for node $u_a$, if $\gamma$ is closed at $key(u_a)$, then we add value $c$ to $\eta$ of $u_a$. 
    For every other node $u$, we find its child $u'$ of key larger than $key(u_a)$ and add $c$ to $z(u')$ (in that the key of every leaf in $\Gamma_{u'}$ is in range $\gamma$). We then perform a push-down operation at each child of $u$ followed by setting attribute $\eta$ of $u$ as the minimum $\eta$ of its children. 
    Note that the current value $\eta$ of $u$ is the minimum $\eta$ 
    over subtree $\Gamma_{u}$ as all previous additions including value $c$ have been properly 
    made to relevant nodes in $\Gamma_{u}$ since values $\eta$ of $u$'s child nodes 
    are up to date after the above push-down operations at them. 

    For $u_b$, we perform a similar routine in the subtree $\Gamma_{u_R}$. After that, we traverse back to the root 
    from $lcs(u_a,u_b)$, during which for each node, we apply a push-down operation at its each child and then reset its attribute $\eta$ as the minimum $\eta$ of its children. 
    Clearly, the running time is $O(\log n)$.

    \item The minimum query. Because the (global) minimum of attributes $\eta$ in $\Gamma$ is always well maintained at its root. The minimum query can be answered in $O(1)$ time. 
\end{itemize}

Recall that our sweeping algorithm processes each event of $h$ 
intersecting line $x=x_j$ for each $x_j\in X$ in $O((|I^j_{in}|+|I^j_{out}|)\log n)$ time. 
Due to $\cup_{t=1}^N|I^t_{in}| = O(n)$ and $\cup_{t=1}^N|I^t_{out}| = O(n)$, including the preprocessing work, we can determine whether there is a point in $B$ hitting all sets $U_i$ in $O(n\log n)$ time. 

\begin{lemma}\label{lem:piercingresult}
    The geometric piercing problem can be solved in $O(n\log n)$ time. 
\end{lemma}

Notice that if there is a point on $S_j$ for some $1\leq j\leq N$ hitting all sets $U_i$, then such a point on $S_j$ can be found in $O(n\log n)$ time. Recall that to decide whether sets $U_i$ can be hit by a point of $B$, 
it is equivalent to determine whether the measure of the union of complement polygons 
of sets $U_i$ in box $B$ is equal to $B$'s. Thus, our algorithm can be adapted to solve this union-measure problem. 

\subsubsection{Wrapping things up}
Recall that any local feasibility test, which is to determine for two given edges $e_1,e_2$ of $G$ whether two points $q_1\in e_1$ and $q_2\in e_2$ exist so that $\max_{i=1}^{k}\phi_i(q_1,q_2)\leq\lambda$, can be reduced to an instance of the above geometric piercing problem. If a point is returned (which hits all sets $U_i$), then $\lambda$ is feasible and $q_1$ (resp., $q_2$) is at the perpendicular projection of that point onto $e_1$ (resp., $e_2$) in the $x$-axis (resp., $y$-axis). Otherwise, $\lambda$ is infeasible w.r.t. edges $e_1,e_2$. Because computing the rectangle set $U_i$ for each pair $P_i\in\calP$ (i.e., the reduction) can be done in $O(n)$ time with the provided distance matrix of $G$. 
We have the following lemma. 

\begin{lemma}
    The local feasibility test can be solved in $O(n\log n)$ time. 
\end{lemma}

To decide the feasibility of $\lambda$, we perform a local feasibility test for every pair of edges in $E$ 
including pairs each containing same edges. If $\lambda$ is infeasible w.r.t. every pair of edges, then $\lambda$ is infeasible for the problem. Thus, the running time of the feasibility test is $O(m^2n\log n)$. 
We finish the proof of Lemma~\ref{lem:graphfeasibility}. 

\section{The problem on a tree}\label{sec:tree}
In this section, we present a faster algorithm for the problem on 
a vertex-weighted tree graph, and then give a linear-time result for 
the unweighted tree version. 

In the following, we use $T$ to represent the given tree. 
Assume that the distance matrix of $T$ is not given. 
To realize a constant-time computation for distance, as preprocessing, we construct in 
linear time the range-minima data structure~\cite{ref:BenderTh00,ref:BenderTh04} for $T$ so that given any two vertices, their lowest common ancestor can be known in constant time. Then, we compute and maintain the distance between the root and every vertex of $T$. 
With these information, the distance between any two points of $T$ can be 
obtained in $O(1)$ time. 

We first present our feasibility test that is to determine for any given $\lambda$, 
whether there exist two points $q_1$ and $q_2$ on $T$ with 
$\max_{i=1}^{k}\phi_i(q_1,q_2)\leq \lambda$. 

\subsection{The feasibility test}\label{subsec:treefeasibility}
For any point $q$ of $T$, removing $q$ from $T$ generates several disjoint subtrees and each of them is a \textit{split} subtree of $q$. 
Additionally, $q$ and each of its split subtrees induces a \textit{hanging} 
subtree of $q$ where $q$ is considered as a virtual point. The terminologies are illustrated in Fig.~\ref{fig:splitsubtree}.

Our algorithm traverses $T$ to place two centers $q_1$ and $q_2$ such that 
$q_1$ (resp., $q_2$) covers all vertices in a hanging subtree of $q_1$ 
(resp., $q_2$) which cannot be covered by any point out of this subtree, where the two hanging subtrees of $q_1$ and $q_2$ have no common vertices. 
Our proof shows that $\lambda$ is feasible iff such two centers lead an 
objective value no more than $\lambda$ . 

Set any leaf of $T$ as its root. At first, we attach an attribute $z(u)$ with 
each vertex $u$ where $z(u)$ is the distance of the furthest point to $u$ 
not in subtree $T_u$ that covers (all vertices of) $T_u$, 
and initialize it as zero. 
Then we traverse $T$ in the post-order from its root to place the first center as follows. 

For each vertex $u$, assuming $u$'s children are $c_1, \cdots, c_s$, 
we update $z(u)$ as the minimum of $\lambda/w(u)$ and values
$z(c_i) - l(e(c_i,u))$ for all $1\leq i\leq s$. 
Note that $z(c_i) - l(e(c_i,u))>0$ for each $1\leq i\leq s$ 
since no centers have been placed yet. 
Now, $z(u)$ is the distance of the furthest point to $u$ not 
in $T_u$ that covers $T_u$. 

We then determine whether a center must be placed on edge $e(u,u')$ 
between $u$ and its parent $u'$. Clearly, if $z(u)-l(e(u,u'))>0$ then 
then a point outside of $T_{u'}$ can cover $T_u$ so we continue our traversal; 
otherwise, we place a center on $e(u, u')$ at the point at 
distance $z(u)$ to $u$ to cover $T_{u}$ and most vertices in $T/T_{u}$

Once the first center is placed on $T$, assuming it is $q_1$, we continue our traversal while determining for each vertex $v$ whether it can be covered by $q_1$ instead. If yes, then $\lambda$ is feasible due to $\max_{i=1}^{k}\phi_i(q_1,q_1)\leq\lambda$. 
Otherwise, at least two centers are needed to cover all vertices of $T$ under $\lambda$. 

Suppose $q_1$ is on the edge $e(\alpha_1, \beta_1)$ 
so that $q_1$ is known to cover $T_{\alpha_1}$ (rather than $T/T_{\alpha_1}$). 
Proceed to traverse subtree $T/T_{\alpha_1}$ starting from its root $\beta_1$ to place the second center $q_2$ as the above. Clearly, the running time for placing $q_1$ and $q_2$ is $O(n)$. 

Suppose $q_2$ is on the edge $e(\alpha_2, \beta_2)$ and it is known to cover $T_{\alpha_2}$. 
Now we visit each pair $P_i\in\calP$ to determine whether 
$\max_{i=1}^{k}\phi_i(q_1,q_2)\leq\lambda$. If yes, $\lambda$ is feasible. 
Otherwise, $\max_{i=1}^{k}\phi_i(q_1,q_2)>\lambda$, and by the following lemma, $\lambda$ is infeasible. 

\begin{observation}\label{obs:decisioncorrectness}
    If $\max_{i=1}^{n}\phi_i(q_1, q_2)>\lambda$, then $\lambda$ is infeasible. 
\end{observation}
\begin{proof}
    By the definition, there must be a pair, say $P_{i'}=(u_{i'},v_{i'})$, 
    so that $\phi_{i'}(q_1,q_2)>\lambda$. 
    Based on the locations of $u_{i'}$ and $v_{i'}$, we prove the truth of the statement 
    for each case as follows. Note that it is not likely that $u_{i'}\in T_{\alpha_1}$ 
    and $v_{i'}\in T_{\alpha_2}$ since otherwise, $\phi_{i'}(q_1,q_2)\leq\lambda$. 

    On the one hand, $u_{i'}$ and $v_{i'}$ are in the same subtree of $T_{\alpha_1}$ and $T_{\alpha_2}$. Recall that $T_{\alpha_1}\cap T_{\alpha_2} = \emptyset$. Suppose they belong to $T_{\alpha_1}$. Then, $q_2$ covers neither of the two vertices. Let $v'$ be any vertex of $T_{\alpha_2}$ that determines $q_2$. Clearly, no point of $T$ can cover both $v'$ and $u_{i'}$ (resp., $v_{i'}$) under $\lambda$. 
    Hence, $\lambda$ is infeasible. The proof is similar for the case where 
    $u_{i'},v_{i'}\in T_{\alpha_2}$.
       
    On the other hand, either $u_{i'}$ or $v_{i'}$ is in subtree $T/(T_{\alpha_1}\cup T_{\alpha_2})$. 
    Due to $\phi_i(q_1,q_2)>\lambda$, the one of them that is in $T/(T_{\alpha_1}\cup T_{\alpha_2})$ can be covered by neither of $q_1$ and $q_2$ under $\lambda$. 
    So, at least three centers are needed to cover all vertices under $\lambda$. 
    It implies that $\lambda$ is infeasible.  
    
    Thus, the observation holds. \qed
\end{proof}

Recall that such two centers can be found in $O(n)$ time; the objective value w.r.t. them can be computed in linear time. Thus, the following lemma is derived.  

\begin{lemma}\label{lem:treefeasibility}
    The feasibility test can be done in $O(n)$ time. 
\end{lemma}

\subsection{The algorithm}\label{subsec:treealgorithm}
Our algorithm consists of two phases: In the first phase, we find the two edges 
$e(\alpha_1^*, \beta_1^*)$ and $e(\alpha_2^*, \beta_2^*)$ that respectively 
consist of $q^*_1$ and $q^*_2$. In the second phase, we compute $\lambda^*$ 
as well as centers $q^*_1$ and $q^*_2$. 

\paragraph{Phase 1.} Compute the \textit{centroid} $c$ of $T$ in linear time, 
which is a node of $T$ so that its hanging subtrees each with size no more than 
$\frac{2}{3}$ of $T$'s. Apply Lemma~\ref{lem:binarysearch} to $c$ to find in linear time 
the hanging subtrees of $c$ that contains $q^*_1$ and $q^*_2$, respectively. The proof of 
Lemma~\ref{lem:binarysearch} is given below. We recursively apply 
Lemma~\ref{lem:binarysearch} to the subtree containing $q^*_1$ (resp., $q^*_2$) 
to find edge $e(\alpha_1^*, \beta_1^*)$ (resp., $e(\alpha_2^*, \beta_2^*)$). 
Clearly, the two edges can be obtained after $O(\log n)$ recursive steps; 
so the running time is $O(n\log n)$.  

\begin{lemma}\label{lem:binarysearch}
    Given any vertex $u$ of $T$, we can find its hanging subtrees respectively containing $q^*_1$ and $q^*_2$ in $O(n)$ time. 
\end{lemma}
\begin{proof}
    Let $T_1, T_2, \cdots, T_s$ represent $u$'s hanging subtrees. 
    Note that $u$ is a virtual vertex in each $T_i$. 
    Define $\tau_i$ as $\max_{v\in T_i}w(v)d(v,u)$ for each $1\leq i\leq s$. 
    Without loss of generality, suppose $\tau_1\geq\tau_2\geq\cdots\geq\tau_s$. 
    Clearly, $\lambda^*\leq\tau_1$.  
    
    We claim that any $T_i$ with $\tau_i<\tau_2$ does not contain centers. 
    Suppose $T_i$ is such a subtree but contains a center. 
    If no vertex in $T_1\cup T_2$ is assigned to the center in $T_i$, then $\lambda^*\geq\tau_2$; 
    if so, this center in $T_i$ can be moved to $u$ without increasing the objective value. 
    Hence, the claim is true. 
    
    By the claim, we consider only the following cases for locating $q^*_1$ and $q^*_2$. 

    \begin{itemize}
        \item $\tau_1 >\tau_2 = \tau_3$. Clearly, a center, say $q^*_1$, must be placed in $T_1$ 
        in order to achieve an objective value smaller than $\tau_1$. We claim $\lambda^*\geq\lambda_2$. For any subtree $T_i$ with $\tau_i=\tau_2$, 
        if one of its vertices that determine $\tau_i$ is assigned to $q^*_1$, then $\lambda^*>\tau_2$. 
        Otherwise, vertices in each subtree $T_i$ with $\tau_i=\tau_2$ that determine $\tau_i$ 
        are all assigned to $q^*_2$, so $\lambda^*\geq\tau_2$. So, the claim is true. 
        
        Clearly, $q^*_2$ must be placed in $T_1\cup\{u\}$ in order to minimize its maximum 
        distance to $T_1$'s vertices. Note that by applying the feasibility test to $\lambda=\tau_2$, 
        we can determine whether $\lambda^* = \tau_2$ or not in linear time. 
        
        \item $\tau_1 = \tau_2 = \tau_3$: If a vertex in $T_i$ with $\tau_i = \tau_1$ that determines 
        $\tau_i$ is covered by a center not at $u$, then the objective value must be 
        larger than $\tau_1$. Clearly, to minimize the objective value, 
        two centers can be placed at $u$, which means $\lambda^* =\tau_1$.  
       
        \item $\tau_1 = \tau_2 > \tau_3$: Suppose there is a pair where its vertices are all in $T_1$ (resp., $T_2$), 
        and their minimum (weighted) distance to $u$ is $\tau_1$. Then, we have 
        $\lambda^* = \tau_1$. It is because a vertex in such a pair must be covered 
        by the same center together with a vertex in $T_2$ (resp., $T_1$) that determines $\tau_2$ 
        (resp., $\tau_1$). Thus, both centers can be placed at $u$. Clearly, the existence of such a pair can be determined in linear time.
        
        In general, there is no such pairs. In order to achieve a smaller objective 
        value than $\tau_1$, $q^*_1$ (resp., $q^*_2$) must be placed in $T_1$ (resp., $T_2$) and 
        the vertices of $T_1$ (resp., $T_2$) leading $\tau_1$ (resp., $\tau_2$) 
        must be assigned to $q^*_1$ (resp., $q^*_2$).  
        
        \item $\tau_1 > \tau_2 > \tau_3$: By the above claim, no centers are in $T_i$ with $i>2$. 
        If there is a pair where its vertices are in $T_1$ and their minimum (weighted) distance 
        to $u$ equals $\tau_1$, then both centers are in $T_1$ due to $\lambda^*\leq\tau_1$. 
        Otherwise, no such pairs exist, so at least one center, say $q^*_1$, must be placed 
        in $T_1$ to cover at least those vertices of $T_1$ determining $\tau_1$.  
       
        Moreover, if $T_2$ contains a pair where the minimum (weighted) distance 
        of its vertices to $u$ equals $\tau_2$, then $\lambda^*>\tau_2$ 
        since one of its vertices is covered by $q_1^*$. If so, both centers can be placed in $T_1$ 
        to minimize the maximum (weighted) distance of vertices in $T_1$ to $q^*_2$. 
        
        Otherwise, no such pairs are in $T_2$. A feasibility test is then 
        applied to $\tau_2$. If $\lambda^*\leq\tau_2$ then $q^*_2$ is 
        in $T_2\cup \{u\}$. Otherwise, $\lambda^*>\tau_2$, so both centers 
        are in $T_1$. 
    \end{itemize}

    Consequently, for each above case, the hanging subtrees of $u$ respectively containing $q^*_1$ and $q^*_2$ can be determined in linear time. Additionally, values $\tau_i$ for all $1\leq i\leq s$ 
    can be determined in $O(n)$ time by traversing each hanging subtree of $u$. Thus, the lemma holds. \qed
\end{proof}

\paragraph{Phase 2.} We form $y = w(u)d(u,x)$ for each vertex $u$ respectively with respect to $x\in e(\alpha_1^*, \beta_1^*)$ and $x\in e(\alpha_2^*, \beta_2^*)$ in the $x,y$-coordinate plane, which can be done in $O(n)$ time. 
As stated in Section~\ref{sec:preliminary}, $\lambda^*$ belongs to the set of $y$-coordinates 
of intersections between the $O(n)$ distance functions. With the assistance of our feasibility test, Lemma~\ref{lem:linearrangement} can be applied to find that intersection in $O((n+\tau)\log n)$ time, 
where $\tau$ is the time of our feasibility test. Hence, we have the following theorem. 

\begin{theorem}\label{the:weightedtree}
    The bichromatic two-center problem on trees can be solved in $O(n\log n)$ time. 
\end{theorem}

\subsection{The unweighted case}\label{subsec:unweightedtree}
In the unweighted case, the weights of vertices in the given tree $T$ are all same. Without loss of generality, we assume that each vertex is of weight one. 

Recall that some vertices of the given tree may not be in any given pair, and for the weighted version, we first preprocess the input graph to reduce the problem into an instance where every vertex of the obtained graph is paired. Instead, for the unweighted tree version, we preprocess the given tree as follows: Traverse the given tree in the post-order to iteratively remove its leaves until its every leave is in a pair of the given set $\calP$. This generates a subtree of $T$ of size no more than $n$ where every leave is in a pair of $\calP$. Clearly, solving the problem on this obtained subtree w.r.t. $\calP$ is equivalent to solving the problem on the given tree w.r.t. $\calP$. With a little abuse of notation, let $T$ represent the obtained subtree. 

Let $q^*$ be the center of $T$ and $\epsilon^*$ be the optimal one-center objective value. As proved in~\cite{ref:HalfinOn74}, $q^*$ is 
the (common) midpoint of the longest path(s) of $T$ and $\epsilon^*$ is the longest path length. Since any longest path of $T$ is between two leaves and every leaf is in a pair of $\calP$, $\lambda^*\leq\epsilon^*$. Additionally, we have the following observation. 

\begin{observation}\label{obs:unweightedcenter}
    $q^*_1$ (resp., $q^*_2$) is at the (common) midpoint of the longest path(s) 
    between two vertices assigned to $q^*_1$ (resp., $q^*_2$); 
    $2\lambda^*$ equals the maximum length of the two longest paths 
    respectively for $q^*_1$ and $q^*_2$. 
\end{observation}
\begin{proof}
    Let $T'$ be the minimum subtree of $T$ that contains all vertices assigned to $q^*_1$ 
    in an optimal solution. Every leaf of $T'$ is assigned to $q^*_1$ but an internal node 
    might be assigned to $q^*_2$ or not in any pair of $\calP$. 
    Each longest path of $T'$ is between two leaves; 
    if there are multiple longest paths then these paths must share the same midpoint. 
    Clearly, the maximum distance from vertices assigned to $q^*_1$ to this (common) midpoint 
    is minimized over $T'$ (or $T$) since the distance from an endpoint of a longest path to 
    any other point of $T'$ (or $T$) is larger than half of the longest path. Hence, 
    $q^*_1$ is at this midpoint, and any vertices of $T'$ except for 
    the endpoints (leaves) of its longest paths are not relevant to $q^*_1$. 
    
    The proof for $q^*_2$ is similar. It follows that $2\lambda^*$ equals the maximum length of 
    two longest paths respectively determining $q^*_1$ and $q^*_2$. 
    Thus, the observation holds. \qed 
\end{proof}

We say that a vertex is a \textit{key} vertex of $q^*$ if its distance to $q^*$ 
equals to $\epsilon^*$ (i.e., if it is an endpoint of the longest path of $T$). 
Every key vertex of $q^*$ is a leaf of $T$, and 
$q^*$ has at least two hanging subtrees each of which contains 
at least one of its key vertices. 

As explained in the proof of Lemma~\ref{lem:binarysearch}, $\lambda^* = \epsilon^*$ 
if either of the two conditions satisfies: (1) $q^*$ has more than two hanging subtrees 
such that each of them contains a key vertex of $q^*$; (2) a hanging subtree of $q^*$ 
contains a pair in $\calP$ whose both vertices are key vertices of $q^*$. 
In general, neither of two conditions satisfies so $\lambda^*<\epsilon^*$. 

Denote by $T_1$ and $T_2$ the two hanging subtrees of $q^*$ containing its key vertices. 
Due to $\lambda^*<\epsilon^*$, $T_1$ and $T_2$ each contains a center; 
assume $q^*_1\in T_1$ and $q^*_2\in T_2$. Let $T_1^*$ be the minimum subtree of $T$ 
induced by all vertices assigned to $q^*_1$ in an optimal solution, and 
$T_2^*$ be the minimum subtree by all vertices assigned to $q^*_2$. 
We have the following observation. 

\begin{observation}\label{obs:keyvertex}
   Each longest path of $T_1^*$ (resp., $T_2^*$) ends with a key vertex of $q^*$ in $T_1$ 
   (resp., $T_2$) and each key vertex of $q^*$ in $T_1$ (resp., $T_2$) 
   is an endpoint of a longest path of $T_1^*$ (resp., $T_2^*$). 
\end{observation}
\begin{proof}
    Below we focus on the proof for $T_1^*$ since the proof for $T^*_2$ is similar. 
    Due to $\lambda^*<\epsilon^*$, all key vertices of $q^*$ in $T_1$ must 
    be assigned to $q^*_1\in T_1$ so they are leaves of $T^*_1$. 
    
    Assume $q^*$ is the root of $T$. Let $u'$ be the lowest common ancestor 
    of all key vertices of $q^*$ in $T_1$. Note that each key vertex of $q^*$ in $T_1$ is 
    of the same distance to $u'$ (since they are the furthest vertices of $q^*$), 
    and $u'$ is in $T^*_1$. 
    Clearly, the longest path of $T^*_1$ is at least as long as the path(s) between $q^*$'s 
    key vertices in $T_1$. Denote by $\delta$ the path length between $q^*$'s key vertices in $T_1$. 
    Note that $\delta = 0$ iff. $q^*$ has only one key vertex in $T_1$. 
    
    Suppose the longest path of $T^*_1$ is of length equal to $\delta$, which means that $q^*$ has 
    multiple key vertices in $T_1$ (or $T_{u'}$). So, $u'$ must be the common midpoint of all longest paths of $T^*_1$. Hence, each key vertex of $q^*$ in $T_1$ is an endpoint of a longest path of $T_1^*$. 
    In addition, each key vertex of $q^*$ in $T_{u'}$ is further to $u'$ than any vertex of $T_{u'}$ that is not a key vertex of $q^*$; the path of 
    any two vertices in $T/T_{u'}$ does not contain $u'$. So, each longest path of $T_1^*$ must end with a key vertex of $q^*$ in $T_{1}$ ($T_{u'}$). Thus, the statement is true. 
    
    Suppose the longest path of $T^*_1$ is of length larger than $\delta$. This implies that each longest path of $T^*_1$ must end with a vertex in $T/T_{u'}$. 
    Let $v$ be any vertex of $T^*_1$ in $T/T_{u'}$. 
    We claim that each key vertex of $q^*$ in $T_{u'}$ (or $T_1$) is $v$'s furthest vertex in $T_1^*$. 
    If $v$ is $T/T_1$, then each key vertex of $q^*$ in $T_{u'}$ is $v$'s furthest vertex in $T$ and thereby, in $T^*_1$. 
    Otherwise, $v$ is $T_1/T_{u'}$. If each key vertex of $q^*$ in $T_{u'}$ is not the furthest vertex 
    of $v$ in $T^*_1$, then $v$'s furthest vertex in $T^*_1$ is not in $T_{u'}$. If so, the path 
    between $v$'s furthest vertex and each key vertex of $q^*$ in $T_1$ 
    is longer than $v$'s path to its furthest vertex. Hence, the claim is true. This implies that the statement is true. 

    It should be mentioned that for either of the above cases, 
    $q^*_1$ must be on the path between $u'$ and $q^*$ since the key vertices of 
    $q^*$ in $T_1$ are furthest to $q^*$ than any other vertex of $T^*_1$. 
    
    For $T^*_2$, the proof is similar so we omit the details. 
    Thus, the observation holds. \qed
\end{proof}

Let $\alpha$ (resp., $\beta$) be any key vertex of $q^*$ in $T_1$ (resp., $T_2$). 
Observation~\ref{obs:unweightedcenter} and Observation~\ref{obs:keyvertex} 
imply the following lemma. 

\begin{lemma}\label{lem:longestpathpartition}
There is an optimal solution where the assignment for each pair $P_i$ is 
determined by $\phi_{i}(\alpha, \beta)$ so that $q^*_1$ (resp., $q^*_2$) is the 
center of all vertices assigned to $\alpha$ (resp., $\beta$). 
\end{lemma}
\begin{proof}
    By Observation~\ref{obs:unweightedcenter} and Observation~\ref{obs:keyvertex}, 
    our problem indeed asks for a partition on each pair $(v_i,u_i)$ 
    so as to minimize the maximum length of the longest paths respectively 
    to any key vertex of $q^*$ in $T_1$ and any key vertex of $q^*$ in $T_2$. Thus, the lemma holds. \qed
\end{proof}

Now we are ready to present our algorithm: First, traverse $T$ in the post-order while iteratively removing its leaves that are not in any pair of $\calP$, which can be done in $O(n)$ time. Second, compute $q^*$ and $\epsilon^*$ in $O(n)$ time~\cite{ref:HalfinOn74}. 
Then, apply Lemma~\ref{lem:binarysearch} to determine whether $\lambda^*=\epsilon^*$. 
If yes, then $q^*_1$ and $q^*_2$ are at $q^*$ and for each pair $P_i\in\calP$, 
we assign any one of its vertices to $q^*_1$ and the other to $q^*_2$. 
Otherwise, we obtain exactly two hanging subtrees $T_1$ and $T_2$ of $q^*$, which contain key vertices of $q^*$. 

Pick a key vertex of $q^*$ from each of the two subtrees; 
denote them by $\alpha$ and $\beta$. 
Proceed to determine in $O(1)$ time the optimal assignment of vertices 
in each pair by letting $q_1 = \alpha$ and $q_2 = \beta$. 
Last, we find in linear time the longest path between $\alpha$ 
and a vertex assigned to $q_1$ as well as the longest path between 
$\beta$ and a vertex assigned to $q_2$. 
$\lambda^*$ equals the half of the maximum length of the two longest paths, and $q^*_1$ and $q^*_2$ are respectively 
at their midpoints. Thus, we have the following theorem. 

\begin{theorem}\label{the:unweightedtree}
    The unweighted bichromatic two-center problem on trees can be solved in $O(n)$ time. 
\end{theorem}

\section{Conclusion} 
The weighted two-center problem on trees can be solved in linear time by a prune-and-search 
algorithm~\cite{ref:Ben-MosheAn06}. However, our result for the bichromatic version 
on a vertex-weighted tree is $O(n\log n)$. So, can we adapt the prune-and-search mechanism 
or extend our linear-time approach for the unweighted case to solve our weighted case 
in linear time? 

Recall that our Lemma~\ref{lem:binarysearch} finds one or two (disjoint) subtrees 
that contain centers. To apply the prune-and-search mechanism, 
we need to determine the optimal assignment for some pairs whose both vertices are 
not in the subtree(s) containing centers in order to find and prune vertices 
assigned to $q^*_1$ that are not relevant to $q^*_1$. 
Let $(v_i,u_i)$ be such a pair. Even though $v_i$ is known to have a larger 
weighted distance at either center than $u_i$, without knowing $q^*_1$ and $q^*_2$, 
it is not clear which assignment for two vertices leads a smaller objective value. 

In addition, since the key vertices of the weighted center of the tree 
might not be relevant to $\lambda^*$, we cannot apply Lemma~\ref{lem:longestpathpartition} 
to determine the optimal assignment for each pair. 
Hence, it seems not likely to extend the prune-and-search scheme or our linear-time framework 
for the unweighted case to solve the weighted version. An open question is whether 
one can solve the bichromatic two-center problem on a vertex-weighted tree in linear time.

%\bibliographystyle{splncs04}
%\bibliography{library.bib}

\end{document}